\documentclass[preprint,12pt]{elsarticle}

\usepackage{amssymb}
\usepackage{amsmath}

\journal{European Journal of Operational Research}

\newtheorem{theorem}{Theorem}

\newtheorem{Example}[theorem]{Example}

\newtheorem{lemma}[theorem]{Lemma}
\newtheorem{proposition}[theorem]{Proposition}

\newenvironment{proof}[1][Proof]{\textbf{#1.} }{\ \rule{0.5em}{0.5em}}

\begin{document}

\begin{frontmatter}



\title{Hedging Conditional Value at Risk with Options}


\author{Maciej J. Capi\'nski
\footnote{Tel.: +48 505429347, Fax: +48 126173165, E-mail: maciej.capinski@agh.edu.pl}}

\address{AGH University of Science and Technology, Al. Mickiewicza 30, 30-059 Krak\'ow, Poland}

\begin{abstract}
We present a method of hedging Conditional Value at Risk of a position in stock using put options. The result leads to a linear programming problem that can be solved to optimise risk hedging. 
\end{abstract}


\end{frontmatter}



\section{Introduction}

One of the natural ideas to reduce risk of a position in stock is to buy put
options. By doing so one can cut off the undesirable scenarios, while
leaving oneself open to the positive outcomes. A choice of a high strike
price of the put option does cut off more of the unfavourable states, but at
the same time produces higher hedging costs. The question of how to balance
the two trends so that the level of risk measured by Value at Risk ($\mathrm{%
VaR}$) is minimised was investigated by Ahn, Boudoukh, Richardson and Whitelaw 
\cite{Ahn}.

The Value at Risk, which is the worst case scenario of loss an investment
might incur at a given confidence level, has established its position as one
of the standard measures of risk, and is widely used throughout the field of
finance and risk management. One of its shortcoming is that it neglects
potential severity of unlikely events. Another, that it is not sub-additive,
and is thus not a coherent risk measure \cite{Delbaen}. Its most common modification to
achieve these goals is the Conditional Value at Risk ($\mathrm{CVaR}$)
(also referred to as `Expected Shortfall'), which takes into the account the
average loss exceeding $\mathrm{VaR}$. The $\mathrm{CVaR}$ is a coherent risk measure (the proof can be found in the work of Acerbi and Tasche \cite{Acerbi}).

In this paper we show a mirror result to \cite{Ahn}, using $\mathrm{CVaR}$
instead of $\mathrm{VaR}$. It turns out that in such setting one can achieve
closed form formulae for $\mathrm{CVaR}$ of stock hedged with puts. These
can be used to optimise the position by solving a linear programming problem.

We restrict our attention to the Black--Scholes model and consider investments in stock and put options.
The optimisation of $\mathrm{CVaR}$ can be carried out under more general
assumptions, using also other securities (as an example see Rockafellar and Uryasev \cite{Rock1, Rock2}). One can also hedge $%
\mathrm{CVaR}$ dynamically (as in the work of Melnikov and Smirnov \cite{Meln}), which provides slightly better results.
Dynamic strategies though require constant rebalancing, which in practice can be costly. Advantages of our approach are as follows: 
its simplicity; closed form analytic formula
for $\mathrm{CVaR}$; protection against risk is very similar to the one attainable using dynamic strategies.

The paper is organised as follows. Section \ref{sec:VaR} recalls the results of Ahn, Boudoukh, Richardson and Whitelaw \cite{Ahn} for hedging of $\mathrm{VaR}$ with put options. This section serves also as preliminaries to the paper. In Section \ref{sec:CVaR} we generalise the result to use $\mathrm{CVaR}$ instead of $\mathrm{VaR}$. The main result of the paper is given in Theorem \ref{th:main}. The section ends with an example of its application. In Section \ref{sec:dynamic} we compare  our method to the results attainable using dynamic strategies. They turn out to be close. We finish the paper with a short conclusion in Section \ref{sec:conclusion}. 

\section{Hedging Value at Risk}\label{sec:VaR}

In this section we set up our notations and recall the results of Ahn,
Boudoukh, Richardson and Whitelaw \cite{Ahn}.

Let $X$ be a random variable, which represents a gain from an investment. For $\alpha $ in $(0,1),$ we define the \emph{Value at Risk} of $X,$ at
confidence level $1-\alpha ,$ as $\mathrm{VaR}^{\alpha }(X)=-q^{\alpha }(X),$
where $q^{\alpha }(X)$ is the upper $\alpha $-quantile of $X$.

We consider the Black--Scholes model, where the stock price evolves
according to $dS(t)=\mu S(t)dt+\sigma S(t)dW(t)$, with the money market
account $dA(t)=rA(t)dt$. A European put option with strike price $K$ and
maturity $T$ has payoff $P(T)=\left( K-S(T)\right) ^{+}$ and costs%
\begin{equation}
P(0)=P(r,T,K,S(0),\sigma)=Ke^{-rT}N(-d_{-})-S(0)N(-d_{+}),  \label{eq:put-price}
\end{equation}%
where%
\begin{eqnarray}
d_{+} &=&d_{+}(r,T,K,S(0),\sigma )=\frac{\ln \frac{S(0)}{K}+\left( r+\frac{1%
}{2}\sigma ^{2}\right) T}{\sigma \sqrt{T}},  \label{eq:d1-d2} \\
d_{-} &=&d_{-}(r,T,K,S(0),\sigma )=d_{+}-\sigma \sqrt{T},  \notag
\end{eqnarray}%
and $N$ is the standard normal cumulative distribution function.

Assume that we buy $x$ shares of stock and $z_{i}$ put options with strikes $%
K_{i},$ which cost $P_{i}(t)$ for $i=1,\ldots ,n$ and $t=0,T$. Let $\mathbf{%
z,}$ $\mathbf{1}$ and $\mathbf{P}(t)$ be vectors in $\mathbb{R}^{n}$ defined
as%
\begin{equation*}
\mathbf{z}=\left[ 
\begin{array}{c}
z_{1} \\ 
\vdots \\ 
z_{n}%
\end{array}%
\right] ,\qquad \mathbf{1}=\left[ 
\begin{array}{c}
1 \\ 
\vdots \\ 
1%
\end{array}%
\right] ,\qquad \mathbf{P}(t)=\left[ 
\begin{array}{c}
P_{1}(t) \\ 
\vdots \\ 
P_{n}(t)%
\end{array}%
\right] .
\end{equation*}%
The value of our investment at time $t$ is $V_{(x,\mathbf{z})}(t)=xS(t)+%
\mathbf{z}^{T}\mathbf{P}(t).$ The following theorem can be used to compute $%
\mathrm{VaR}$ for the discounted gain%
\begin{equation*}
X_{(x,\mathbf{z})}=e^{-rT}V_{(x,\mathbf{z})}(T)-V_{(x,\mathbf{z})}(0).
\end{equation*}

\begin{theorem}
\label{th:VaR}\cite{Ahn} If $z_{i}\geq 0,$ for $i=1,\ldots ,n,$ and $\mathbf{%
z}^{T}\mathbf{1}\leq x,$ then%
\begin{equation}
\mathrm{VaR}^{\alpha }\left( X_{(x,\mathbf{z})}\right) =V_{(x,\mathbf{z}%
)}(0)-e^{-rT}\left( xq^{\alpha }(S(T))-\mathbf{z}^{T}\mathbf{q}^{\alpha }(-%
\mathbf{P}(T))\right) ,  \label{eq:VaR-BS-model}
\end{equation}%
where 
\begin{equation}
\mathbf{q}^{\alpha }(-\mathbf{P}(T))=-\left[ 
\begin{array}{c}
\left( K_{1}-q^{\alpha }(S(T))\right) ^{+} \\ 
\vdots \\ 
\left( K_{n}-q^{\alpha }(S(T))\right) ^{+}%
\end{array}%
\right] .  \label{eq:Quantile-H}
\end{equation}
\end{theorem}

\section{Hedging Conditional Value at Risk}\label{sec:CVaR}

One of the shortcomings of $\mathrm{VaR}$ is that it neglects the tail of
the loss distribution. An improvement in this respect is the \emph{%
Conditional Value at Risk}, defined as%
\begin{equation*}
\mathrm{CVaR}^{\alpha }(X)=\frac{1}{\alpha }\int_{0}^{\alpha }\mathrm{VaR}%
^{\beta }(X)d\beta =-\frac{1}{\alpha }\int_{0}^{\alpha }q^{\beta }(X)d\beta ,
\end{equation*}%
with a well known equivalent form 
\begin{equation}
\mathrm{CVaR}^{\alpha }(X)=-\frac{1}{\alpha }\left[ \mathbb{E}(X\mathbf{1}%
_{\{X\le q^{\alpha }(X)\}})+q^{\alpha }(X)(\alpha -\mathbb{P}(X\le q^{\alpha
}(X))\right] .  \label{Eq:CVaR_jump}
\end{equation}%
The $\mathrm{CVaR}$ also has the advantage of being a coherent risk measure 
\cite{Acerbi, Delbaen}.

Our aim is to give a mirror result to Theorem \ref{th:VaR}, using $\mathrm{%
CVaR}$ as the risk measure. We start with a simple lemma.

\begin{lemma}
\label{lem:conditional-exp-S(T)}For any $q\in \mathbb{R}$, 
\begin{equation*}
\mathbb{E}\left( S(T)|W(T)\leq q\sqrt{T}\right) =\frac{1}{N(q)}S(0)e^{\mu
T}N\left( q-\sigma \sqrt{T}\right) .
\end{equation*}
\end{lemma}

\begin{proof}
Let $Z=W(T)/\sqrt{T}$. Since $\mathbb{P}(Z\leq q)=N(q)>0,$%
\begin{align*}
\mathbb{E}\left( S(T)|Z\leq q\right) & =\frac{1}{P(Z\leq q)}\int_{-\infty
}^{q}S(0)e^{\left( \left( \mu -\frac{\sigma ^{2}}{2}\right) T+\sigma \sqrt{T}%
x\right) }\frac{1}{\sqrt{2\pi }}e^{-\frac{x^{2}}{2}}dx \\
& =\frac{1}{N(q)}S(0)e^{\mu T}\int_{-\infty }^{q}\frac{1}{\sqrt{2\pi }}e^{-%
\frac{\left( x-\sigma \sqrt{T}\right) ^{2}}{2}}dx \\
& =\frac{1}{N(q)}S(0)e^{\mu T}N\left( q-\sigma \sqrt{T}\right) ,
\end{align*}%
as required.
\end{proof}

Let $Z$ be a random variable with standard normal distribution $N(0,1)$. To
compute $\mathrm{CVaR}^{\alpha }(X_{(x,\mathbf{z})})$, we introduce
notations 
\begin{equation*}
\begin{array}{lll}
d_{-}^{\mu }=d_{-}(\mu ,T,K,S(0),\sigma ), &  & d_{+}^{\mu }=d_{-}^{\mu
}+\sigma \sqrt{T}, \\ 
d_{-}^{\mu ,\alpha }=\max \left( d_{-}^{\mu },-q^{\alpha }(Z)\right) , &  & 
d_{+}^{\mu ,\alpha }=d_{-}^{\mu ,\alpha }+\sigma \sqrt{T},%
\end{array}%
\end{equation*}%
\begin{equation}
P^{\alpha }(K)=Ke^{-\mu T}N(-d_{-}^{\mu ,\alpha })-S(0)N\left( -d_{+}^{\mu
,\alpha }\right) .  \label{eq:P-mu-alpha}
\end{equation}
We first consider the case when we invest in puts with a single strike $%
K_{1}=K$.

\begin{proposition}
\label{prop:CVaR-single-strike}If $\mathbf{z}=\left[ z_{1}\right] ,$ for $%
z_{1}=z\in \lbrack 0,x]$, then 
\begin{equation*}
\mathrm{CVaR}^{\alpha }\left( X_{\left( x,\mathbf{z}\right) }\right) =V_{(x,%
\mathbf{z})}(0)-\frac{1}{\alpha }e^{(\mu -r)T}\left[ xS(0)N\left( q^{\alpha
}(Z)-\sigma \sqrt{T}\right) +zP^{\alpha }(K)\right] .
\end{equation*}
\end{proposition}

\begin{proof}
We first observe that%
\begin{equation}
X_{\left( x,\mathbf{z}\right) }=e^{-rT}\left( xS(T)+z\left( K-S(T)\right)
^{+}\right) -V_{(x,\mathbf{z})}(0).  \label{eq:CVaR-put-tmp1}
\end{equation}%
Since $z\leq x$, we see that 
\begin{equation}
s\rightarrow e^{-rT}\left( xs+z\left( K-s\right) ^{+}\right) -V_{(x,\mathbf{z%
})}(0)  \label{eq:s-gain-for-CVaR}
\end{equation}%
is a non-decreasing function of $s$. Also $\xi \rightarrow S(0)\exp \left(
\left( \mu -\sigma ^{2}/2\right) T+\sigma \sqrt{T}\xi \right) $ is
increasing. Combining these two facts, taking $Z=W(T)/\sqrt{T},$%
\begin{equation}
\left\{ X_{(x,\mathbf{z})}\leq q^{\alpha }(X_{(x,\mathbf{z})})\right\}
=\left\{ S(T)\leq q^{\alpha }(S(T))\right\} =\left\{ Z\leq q^{\alpha
}(Z)\right\} .  \label{eq:BS-CVaR-quantiles-put}
\end{equation}

We first prove the claim for $z<x.$ Then (\ref{eq:s-gain-for-CVaR}) is
strictly increasing, therefore $\mathbb{P}(X_{(x,\mathbf{z})}\leq q^{\alpha
}(X_{(x,\mathbf{z})}))=\mathbb{P}(S(T)\leq q^{\alpha }(S(T)))=\alpha ,$ and%
\begin{align}
\mathrm{CVaR}^{\alpha }(X_{\left( x,\mathbf{z}\right) })& =-\mathbb{E}\left(
X_{(x,\mathbf{z})}|X_{(x,\mathbf{z})}\leq q^{\alpha }(X_{\left( x,\mathbf{z}%
\right) })\right)  \notag \\
& =-\mathbb{E}\left( X_{(x,\mathbf{z})}|Z\leq q^{\alpha }(Z)\right) \hspace{%
3.3cm}\text{(see (\ref{eq:BS-CVaR-quantiles-put}))}  \notag \\
& =V_{(x,z)}(0)-e^{-rT}x\mathbb{E}\left( S(T)|Z\leq q^{\alpha }(Z)\right)
\qquad \text{(see (\ref{eq:CVaR-put-tmp1}))}  \notag \\
& \quad -e^{-rT}z\mathbb{E}\left( \left( K-S(T)\right) ^{+}|Z\leq q^{\alpha
}(Z)\right) .  \label{eq:CVaR-Stock-put-tmp}
\end{align}%
We now compute the last term in (\ref{eq:CVaR-Stock-put-tmp}). Since $%
\left\{ S(T)\leq K\right\} =\left\{ Z\leq -d_{-}^{\mu }\right\} ,$%
\begin{align*}
& \mathbb{E}\left( \left( K-S(T)\right) ^{+}|Z\leq q^{\alpha }(Z)\right) \\
& =\frac{1}{\alpha }\int_{-\infty }^{\min (q^{\alpha }(Z),-d_{-}^{\mu
})}\left( K-S(0)e^{\left( \mu -\frac{\sigma ^{2}}{2}\right) T+\sigma \sqrt{T}%
x}\right) \frac{1}{\sqrt{2\pi }}e^{-x^{2}}dx \\
& =\frac{1}{\alpha }\int_{-\infty }^{-d_{-}^{\mu ,\alpha }}K\frac{1}{\sqrt{%
2\pi }}e^{-x^{2}}dx-\frac{1}{\alpha }\int_{-\infty }^{-d_{-}^{\mu ,\alpha
}}S(0)e^{\left( \mu -\frac{\sigma ^{2}}{2}\right) T+\sigma \sqrt{T}x}\frac{1%
}{\sqrt{2\pi }}e^{-x^{2}}dx \\
& =\frac{1}{\alpha }KN(-d_{-}^{\mu ,\alpha })-\frac{1}{\alpha }\mathbb{P}%
(Z\leq -d_{-}^{\mu ,\alpha })\mathbb{E}\left( S(T)|Z\leq -d_{-}^{\mu ,\alpha
}\right) \\
& =\frac{1}{\alpha }KN(-d_{-}^{\mu ,\alpha })-\frac{1}{\alpha }S(0)e^{\mu
T}N\left( -d_{-}^{\mu ,\alpha }-\sigma \sqrt{T}\right) \qquad \text{(by
Lemma \ref{lem:conditional-exp-S(T)})} \\
& =\frac{1}{\alpha }e^{\mu T}\left( Ke^{-\mu T}N(-d_{-}^{\mu ,\alpha
})-S(0)N\left( -d_{+}^{\mu ,\alpha }\right) \right) .
\end{align*}%
Substituting the above into (\ref{eq:CVaR-Stock-put-tmp}) and applying Lemma %
\ref{lem:conditional-exp-S(T)} gives the claim.

We now need to consider the case when $z=x$. Since for any $\beta \in (0,1)$%
, $\lim_{z\nearrow x}q^{\beta }(X_{(x,\mathbf{z})})=q^{\beta }(X_{(x,x)}),$
we obtain 
\begin{multline*}
\lim_{z\nearrow x}\mathrm{CVaR}^{\alpha }\left( X_{(x,\mathbf{z})}\right)
=\lim_{z\nearrow x}\frac{-1}{\alpha }\int_{0}^{\alpha }q^{\beta }(X_{(x,%
\mathbf{z})})d\beta \\
=\frac{-1}{\alpha }\int_{0}^{\alpha }q^{\beta }(X_{(x,x)})d\beta =\mathrm{%
CVaR}^{\alpha }\left( X_{(x,x)}\right) .
\end{multline*}%
Hence the result follows from the fact that the formula for $\mathrm{CVaR}%
^{\alpha }(X_{(x,\mathbf{z})})$ in the claim is continuous with respect to $%
z $.
\end{proof}

We can now formulate our main result.

\begin{theorem}
\label{th:main}If $z_{i}\geq 0$ for $i=1,\ldots ,n$ and $z_{1}+\ldots
+z_{n}\leq x,$ then%
\begin{equation}
\mathrm{CVaR}^{\alpha }(X_{(x,\mathbf{z})})=V_{(x,\mathbf{z})}(0)-\frac{1}{%
\alpha }e^{(\mu -r)T}\left[ xS(0)N\left( q^{\alpha }(Z)-\sigma \sqrt{T}%
\right) +\mathbf{z}^{\mathrm{T}}\mathbf{P}^{\alpha }\right] ,
\label{eq:CVaR-BS-model}
\end{equation}%
where $\mathbf{P}^{\alpha }=\left( P^{\alpha }(K_{1}),\ldots ,P^{\alpha
}(K_{n})\right) $.
\end{theorem}

\begin{proof}
The proof follows from mirror arguments to the proof of Proposition \ref%
{prop:CVaR-single-strike}.
\end{proof}

We show how Theorem \ref{th:main} can be applied. Assume that $x$ is fixed.
We investigate how to minimise $\mathrm{CVaR}^{\alpha }\left( X_{(x,\mathbf{z%
})}\right) $ by choosing $\mathbf{z}$. Assume that we invest $V_{0}$ and
spend $c=V_{0}-xS(0)$ on put options. By (\ref{eq:CVaR-BS-model}),
minimising $\mathrm{CVaR}^{\alpha }\left( X_{(x,\mathbf{z})}\right) $ is
equivalent to the problem: 
\begin{equation}
\begin{array}{ll}
\min -\mathbf{z}^{\mathrm{T}}\mathbf{P}^{\alpha } &  \\ 
\text{subject to:} & \mathbf{z}^{\mathrm{T}}\mathbf{P}(0)=c, \\ 
& \mathbf{z}^{\mathrm{T}}\mathbf{1}\leq x, \\ 
& z_{0},\ldots ,z_{n}\geq 0.%
\end{array}
\label{eq:CVaR-BS-opt-prob}
\end{equation}%
This is a linear programming problem, which can easily be solved numerically.

The result can be complemented by computing $\mathbb{E}\left( X_{(x,\mathbf{z%
})}\right) $ to give risk/return type analysis. A direct computation gives
\begin{equation*}
\mathbb{E}\left( X_{(x,\mathbf{z})}\right) =e^{-rT}\left[ xS(0)e^{\mu T}+%
\mathbf{z}^{\mathrm{T}}\mathbb{E}\left( \mathbf{P}(T)\right) \right]
-V_{(x,z)}(0),
\end{equation*}%
where%
\begin{equation*}
\mathbb{E}\left( \mathbf{P}(T)\right) =e^{\mu T}\left[
\begin{array}{c}
P(\mu ,T,K_{1},S(0),\sigma ) \\ 
\vdots \\ 
P(\mu ,T,K_{n},S(0),\sigma )%
\end{array}%
\right] .
\end{equation*}

\begin{Example}\label{ex:puts}
\label{example}Consider $S(0)=100$, $\mu =10\%,$ $\sigma =0.2$ and $r=3\%$.
Assume that we spend $V_{0}=1000,$ investing in stock and put options with
strike prices $K_{1}=80,$ $K_{2}=90,$ $K_{3}=100,$ $K_{4}=110$, $K_{5}=120$
and expiry $T=1$. We shall solve (\ref{eq:CVaR-BS-opt-prob}) for $\alpha
=0.05,$ considering $c\in \left[ 0,160\right] $.

The choice of $x$ depends on $c$, since $xS(0)+c=V_{0}.$

We compute the vectors:
\begin{equation*}
\mathbf{P}(0)=\left[ 
\begin{array}{c}
0.860 \\ 
2.769 \\ 
6.458 \\ 
12.042 \\ 
19.220%
\end{array}%
\right] ,\qquad \mathbf{P}^{\alpha }=\left[ 
\begin{array}{c}
0.366 \\ 
0.819 \\ 
1.271 \\ 
1.724 \\ 
2.176%
\end{array}%
\right],
\qquad \mathbb{E}(\mathbf{P}(T))=\left[ 
\begin{array}{c}
0.420\\
1.574\\
4.148\\
8.527\\
14.686
\end{array}%
\right] .
\end{equation*}%
The solutions to the problem (\ref{eq:CVaR-BS-opt-prob}) are:%
\begin{equation*}
\begin{tabular}{c c c c c c c c c} \hline\hline
$c$ &	$x$ &	$z_1$ &	$z_2$ &	$z_3$ &	$z_4$ & $z_5$ & $CVaR^{\alpha}$ & $\mathbb{E}$ \\
[0.5ex] \hline
0 &	10 &	0 &	0 &	0 &	0 &	0 &	302.24 &	72.51 \\
20 &	9.8 &	3.74 &	6.06 &	0 &	0 &	0 &	180.35 &	61.84 \\
40 &	9.6 &	0 &	5.96 &	3.64 &	0 &	0 &	126.24 &	53.35 \\
60 &	9.4 &	0 &	0.19 &	9.21 &	0 &	0 &	89.64 &	45.52 \\
80 &	9.2 &	0 &	0 &	5.51 &	3.69 &	0 &	71.42 &	39.41 \\
100 &	9 &	0 & 	0 &	1.50 &	7.50 &	0 &	53.82 &	33.35 \\
120 &	8.8 &	0 &	0 &	0 &	6.85 &	1.95 &	41.64 &	28.31 \\
140 &	8.6 &	0 &	0 &	0 &	3.52 &	5.08 &	32.70 &	23.86 \\
160 &	8.4 &	0 &	0 &	0 &	0.20 &	8.20 &	23.75 &	19.42 \\
[1ex] \hline\end{tabular}%
\end{equation*}%
From the table we observe that for larger $c$ we can afford to buy options
with higher strike prices, which provide better protection, but are at the
same time more expensive.
\end{Example}

\section{Comparison with dynamic hedging}\label{sec:dynamic}
An alternative to hedging with put options is to engage in a self financing strategy that will reduce the risk. In this section we explore the differences between this approach and our method.

F{\"o}llmer and Leukert \cite{Folmer} developed a method for dynamic optimisation of
VaR. In \cite{Meln}, Melnikov and Smirnov (by combining techniques from
\cite{Folmer} with \cite{Rock1,Rock2}) extend the method to the setting of dynamic
optimisation of $\mathrm{CVaR}$. They consider a contingent claim with a time $T$ payoff $H$, and solve the following problem:
\begin{equation}%
\begin{array}
[c]{l}%
\underset{\xi}{\min}\,\mathrm{CVaR}^{\alpha}(e^{-rT}(V_{\xi}(T)-H)),\\
\text{subject to }V_{\xi}(0)\leq V_{0},
\end{array}
\label{eq:CVaR-dyn-prob}%
\end{equation}
where $V_{\xi}(t)$ is the time $t$ value of a self financing strategy $\xi,$ and $V_{0}\leq
\mathbb{E}_{\ast}\left(  H\right) .$ (Here $\mathbb{E}_{\ast}$ stands for expectation with respect to the risk neutral measure.) Problem (\ref{eq:CVaR-dyn-prob}), in other words, is how to
minimise the risk of a position in a contingent claim $H$, having available $V_{0}$ for hedging, which is smaller than the cost of the replicating strategy of the claim.

In our setting, we hedge a position in $x$ shares of stock. We can take 
\begin{equation}
H=e^{rT}V_{0}-xS(T).\label{eq:CVaR-dyn-payoff}%
\end{equation}
The interpretation of such choice of $H$ is as follows. We borrow $V_{0}$ and buy
$x$ shares of stock.  The remaining
\[c=V_{0}-xS(0),\]
is spent on a self financing strategy $\xi$, which involves continuous time trading
in stock and money market account. The combined position at time $T$ is $-e^{rT}V_{0}+xS(T)+V_{\xi}(T)$. After discounting, this is
\[-V_{0}+e^{-rT}xS(T)+e^{-rT}V_{\xi}(T)=e^{-rT}(V_{\xi}(T)-H),\]
which fits the framework of problem (\ref{eq:CVaR-dyn-prob}).

The following theorem provides the solution to problem (\ref{eq:CVaR-dyn-prob}%
) for the payoff (\ref{eq:CVaR-dyn-payoff}). It is a reformulation of Theorem
2.4 from \cite{Meln} (adapted to our particular setting and notations).

\begin{theorem}\cite{Meln}
\label{th:dynamic-CVaR}Let $K^{\ast}\in\mathbb{R}$ be a number satisfying%
\[
c=x\mathbb{E}_{\ast}\left(  e^{-rT}\left(  K^{\ast}-S(T)\right)
^{+}\right)  .
\]
Let $b(K)$ be a function implicitly defined by%
\begin{equation}
c=x\mathbb{E}_{\ast}\left(  e^{-rT}\left(  K-S(T)\right)  ^{+}\mathbf{1}_{\left\{
S_{T}>b(K)\right\}  }\right)  ,\label{eq:CVaR-constraint-xi}%
\end{equation}
and let%
\begin{equation}
\mathfrak{c}(K)=\left\{
\begin{array}
[c]{ll}%
V_{0}-xe^{-rT}K+\frac{xe^{-rT}}{\alpha}\mathbb{E}(\left(  K-S(T)\right)
^{+}\mathbf{1}_{\left\{  S_{T}\leq b(K)\right\}  }) & \quad\text{for }K>K^{\ast
}\\
V_{0}-xe^{-rT}K & \quad\text{for }K\leq K^{\ast}.%
\end{array}
\right.  \label{eq:CVaR-xi}%
\end{equation}
Let $H$ be defined by (\ref{eq:CVaR-dyn-payoff}). Then the solution of problem (\ref{eq:CVaR-dyn-prob}) is
\begin{equation}
CVaR^{\alpha}(e^{-rT}(V_{\xi}(T)-H))=\min_{K}\mathfrak{c}(K),\label{eq:CVaR-minimising}%
\end{equation}
and the optimal strategy is the one replicating the contingent claim with the
payoff%
\begin{equation}
\left(  K-S(T)\right)  ^{+}\mathbf{1}_{\left\{  S_{T}>b(K)\right\}
}.\label{eq:G-hedging}%
\end{equation}

\end{theorem}

Since%
\[
\left(  K-S(T)\right)  ^{+}\mathbf{1}_{\left\{  S_{T}\leq b\right\}  }=\left(
b-S\left(  T\right)  \right)  ^{+}+\left(  K-b\right)  \mathbf{1}_{\left\{  S(T)\leq
b\right\}  },
\]
the term involving expectation in (\ref{eq:CVaR-xi}) is%
\begin{align*}
\mathbb{E}&(\left(  K-S(T)\right)  ^{+}\mathbf{1}_{\left\{  S_{T}\leq b\right\}  }) = \\
&=e^{\mu T}\mathbb{E}\left(  e^{-\mu T}\left(  b-S\left(  T\right)  \right)
^{+}\right)  +\left(  K-b\right)  \mathbb{E}\left(  \mathbf{1}_{\left\{  S(T)\leq
b\right\}  }\right)  \\
&  =e^{\mu T}P\left(  \mu,T,b,S(0),\sigma\right)  +\left(  K-b\right)
N\left(  -d_{-}\left(  \mu,T,b,S(0),\sigma\right)  \right)  .
\end{align*}
Similarly, since%
\[
\left(  K-S(T)\right)  ^{+}\mathbf{1}_{\left\{  S_{T}>b\right\}  }=\left(
K-S(T)\right)  ^{+}-\left(  b-S(T)\right)  ^{+}-\left(  K-b\right)
\mathbf{1}_{\left\{  S(T)\leq b\right\}  },
\]
the constraint (\ref{eq:CVaR-constraint-xi}) is%
\begin{align*}
c    =&\, xP\left(  r,T,K,S(0),\sigma\right)  -xP\left(  r,T,b,S(0),\sigma\right)
\\
&  -x\left(  K-b\right)  e^{-rT}N\left(  -d_{-}\left(  r,T,b,S(0),\sigma
\right)  \right)  .
\end{align*}
This means that we have analytic formulae for all the ingredients of Theorem
\ref{th:dynamic-CVaR}, and thus problem (\ref{eq:CVaR-minimising}) can be
solved numerically with relative ease.

\begin{Example}
\label{ex:dynamic}As in Example \ref{ex:puts}, consider $S(0)=100$, $\mu=10\%$,
$\sigma=0.2$, $r=3\%$ and the hedging costs $c=20,40,60,\ldots,160$. The $K$
solving (\ref{eq:CVaR-minimising}) and the resulting optimal $CVaR^{\alpha
}\left(  e^{-rT}(V_{\xi}(T)-H\right))  $ are as follows:
\begin{center}
\begin{tabular}{c c c c} \hline\hline$c$ &	$x$ & K& $CVaR^{\alpha}$  \\
[0.5ex] \hline20 &	9.8& 87.06 & 172.06  \\
40 &	9.6& 94.43 & 120.23  \\
60 &	9.4& 99.84 & 89.25  \\
80 &	9.2& 104.41 & 67.85  \\
100 &	9 & 108.53 & 52.10  \\
120 &	8.8 & 112.40 & 40.12  \\
140 &	8.6 & 116.12 & 30.84  \\
160 &	8.4 & 119.78 & 23.59  \\
[1ex] \hline\end{tabular}
\end{center}%

\end{Example}

By comparing the values from tables in Examples \ref{ex:puts} and \ref{ex:dynamic}, we
see that optimal $\mathrm{CVaR}^{\alpha}$ from dynamic hedging are close to
$\mathrm{CVaR}^{\alpha}$ for the static hedging with puts. Since the difference is
small, an investor might prefer to buy a portfolio of puts and go for a static
hedging position, rather than engage in a dynamic hedging strategy.

\section{Conclusion}\label{sec:conclusion}
We have provided an analytic solution for $\mathrm{CVaR}$ of a position in stock hedged by put options. 
We have shown that the problem of minimising $\mathrm{CVaR}$ reduces to a linear programming problem that can easily be solved in practice. We have demonstrated that thus obtained results are close to the ones attainable using dynamic hedging.





\end{document}